\documentclass{article}
\usepackage{graphicx}
\usepackage{algorithm2e}
\usepackage{float}

\usepackage[color=orange!25,linecolor=orange!80!black]{todonotes}

\usepackage{amsmath, amssymb} 
\usepackage{amsthm}           

\theoremstyle{definition}
\newtheorem{definition}{Definition}

\theoremstyle{plain}
\newtheorem{theorem}{Theorem}
\newtheorem{proposition}{Proposition}
\newtheorem{lemma}{Lemma}

\usepackage{tikz}
\usepackage{tikz-network}
\usetikzlibrary{arrows,arrows.meta,calc}
\usetikzlibrary{shadows}
\usepackage{enumerate}
\usepackage{xspace}
\usetikzlibrary{shapes.multipart}
\usetikzlibrary{positioning}
\usetikzlibrary{shadows}
\usetikzlibrary{shapes.multipart}

\usepackage{xcolor-solarized}

\usepackage{listings}
\usepackage{appendix}
\lstdefinelanguage{pgkeysLang}
{
  morekeywords={
    NODE,
    EDGE,
    TYPE,
    IMPORTS,
    OPTIONAL,
    DATE,
    INT,
    STRING,
    BOOL,
    ARRAY,
    ENUM,
    DOUBLE,
    INT32,
    DATETIME,
    STRICT,
    LOOSE,
    OPEN,
    CLOSED,
    ABSTRACT,
    CREATE,
    GRAPH,
    FOR,
    WITHIN,
    EXCLUSIVE,
    MANDATORY,
    SINGLETON,
    FLOAT,
    IDENTIFIER,
    MATCH,
    WHERE,
    NOT,
    OR,
    AND,
    EXISTS,
    RETURN,
    IN,
    IS,
    NULL,
    hasKey,
    FunctionalProperty,
    context,
    test,
    unique,
    selector,
    field,
    allInstances,
    isUnique,
    Tuple,
    class,
    extent,
    relationship,
    inverse,
    attribute,
    COUNT,
    OF
  },
  sensitive=false, 
  morecomment=[l]{//}, 
  morecomment=[s]{/*}{*/}, 
  morestring=[b]" 
}

\usepackage{color}
\definecolor{eclipseBlue}{RGB}{42,0.0,255}
\definecolor{eclipseGreen}{RGB}{63,127,95}
\definecolor{eclipsePurple}{RGB}{127,0,85}
 
\usepackage{xcolor}

\usepackage[table]{xcolor}
\definecolor{LightGray}{gray}{0.93}

\usepackage[style=numeric, backend=biber]{biblatex}
\newcommand{\OMIT}[1]{}

\newcommand{\np}{\mathrm{NP}}
\newcommand{\ptime}{\mathrm{PTIME}}

\newcommand{\KK}{\ensuremath{\mathcal{K}}}
\newcommand{\LL}{\ensuremath{\mathcal{L}}}

\newcommand{\RR}{\ensuremath{\mathcal{R}}}

\newcommand{\VV}{\ensuremath{\mathcal{V}}}

\newcommand{\pgschema}{\textsc{PG-Schema}\xspace}
\newcommand{\pgtypes}{\textsc{PG-Types}\xspace}

\newcommand{\pgkeys}{\textsc{PG-Keys}\xspace}

\newcommand{\pto}{\rightharpoonup}

\newcommand{\omitforcameraready}[1]{}

\title{Validation of graph databases against PG-Schema}

\date{May 2026}

\author{
Jacek Ciszewski\\
University of Warsaw \and 
Jakub K{\l}os\\
University of Warsaw \and 
Maxime Jakubowski\\ 
TU Wien \and
Dominik Tomaszuk\\
University of Bia{\l}ystok \and
Filip Murlak\\
University of Warsaw
}

\begin{document}
\maketitle

\begin{abstract}
The problem of validating a given graph database instance against a given \pgschema graph type without integrity constraints is NP-complete in terms of combined complexity and in PTIME in terms of data complexity. The combined complexity drops to PTIME when the alternation between type combinations and unions is suitably restricted.
\end{abstract}

\section{Property Graph Data Model}
\label{sec:property-graph}

We assume countable sets $\LL$, $\KK$, and $\VV$ of \emph{labels}, \emph{property names (keys)}, and \emph{property values}. A \emph{record} with keys from $\KK$ and values from $\VV$ is a finite-domain partial function $o \colon \KK \pto \VV$ mapping keys to values.
We write $|o|$ for the size of the domain of record $o$ and $\RR$ for the set of all records.

\begin{definition}[Property Graph]
A \emph{property graph} is defined as a tuple $G = (N, E, \rho, \lambda, \pi)$ where:
\begin{itemize}
\item 
$N$ is a finite set of nodes; 
\item 
$E$ is a finite set of edges 
such that $N \cap E = \emptyset$;
\item 
$\rho : E \rightarrow (N \times N)$ is a total function mapping edges 
to ordered pairs of nodes (the endpoints of the edge); 
\item 
$\lambda : (N \cup E) \rightarrow 2^{\LL}$ is a total function mapping nodes and edges 
to finite sets of labels (including the empty set);
\item 
$\pi : (N \cup E) \to \RR$ is a total function mapping nodes and edges to records. 
\end{itemize}
\end{definition}

For an edge $e \in E$ with $\rho_G(e) = (u,v)$, the nodes $u$ and $v$ are the \emph{endpoints} of $e$, where  $u$ is the \emph{source} and $v$ is the \emph{target} of $e$.
For an element $x \in N \cup E$, the record $\pi(x)$ collects all \emph{properties} of $x$ (key-value pairs) and is called the \emph{content} of $x$.

\begin{figure}[t]
\centering
    \footnotesize
    \begin{tikzpicture}[
      Node/.style={
        rounded corners,
        rectangle split,
        rectangle split parts=2,
        rectangle split part fill={green!80!black!15!white,solarized-base3!50!white},
        draw=green!20!black,
        semithick
      },
      NodeName/.style=
      {
        text=blue,
      },
      ISA/.style={
        thick,-{Triangle[open]},
      },
      Edge/.style={
        thick,-stealth',
      },
      EdgeName/.style=
      {
        text=blue,
      },
      EdgeLabels/.style=
      {
        sloped,
      },
      EdgeProperties/.style={
        sloped,
        below,
        rounded corners,
        fill=solarized-base3!50!white,
      },
      Junction/.style={
        circle,
        draw=black,
        fill=white,
        thick
      }
      ]

      \node[Node] (person) at (3, 5) {
        \texttt{Person}, \texttt{Customer}
        \nodepart{two}
        \begin{tabular}{l}
        \texttt{name: 'Jan Kowalski'}\\
        \texttt{birthdate: 12.12.13}\\
        \texttt{id: 0000123001}
        \end{tabular}
      };

      \node[NodeName] (personname) at (1, 6) {$u_1$};
      \node[NodeName] (companyname) at (5.7, 5.8) {$u_2$};
      \node[NodeName] (person2name) at (-1, 3.05) {$u_3$};
      \node[NodeName] (accountname) at (2.5, 0.7) {$u_4$};

      \node[EdgeName] (owns1name) at (4, 3.8) {$e_1$}; 
      \node[EdgeName] (owns1name) at (7.9, 4) {$e_2$};

      \node[Node] (company) at (8, 5) {
        \texttt{Company}
        \nodepart{two}
        \begin{tabular}{l}
        \texttt{name: 'Kowalski \& Sons'}\\

        \texttt{VAT\_Number: '826009065'}
        \end{tabular}
      };

      \node[Node] (person2) at (1, 2.3) {
        \texttt{Person}
        \nodepart{two}
        \begin{tabular}{l}
        \texttt{name: 'Jan Jansen'}\\
        \texttt{birthdate: 27.05.97}
        \end{tabular}
      };
      
      \node[Node] (account) at (6, 0) {
        \texttt{Account}
        \nodepart{two}
        \begin{tabular}{l}
        \texttt{iban: '94 0000 0000 1977 1341 1345 2342'}
        \end{tabular}
      };

      \draw[] (person) edge[Edge] 
      node[above,EdgeLabels] {\texttt{Owns}}       node[below,EdgeProperties] {\texttt{since: 01.01.25}}
      (account);

\draw[] (company) edge[Edge] 
      node[above,EdgeLabels] {\texttt{Owns}}       node[below,EdgeProperties] {\texttt{since: 11.11.19}}
      (account);
      
    \end{tikzpicture}
\caption{A small customer graph.}
\label{fig:graph}
\end{figure}
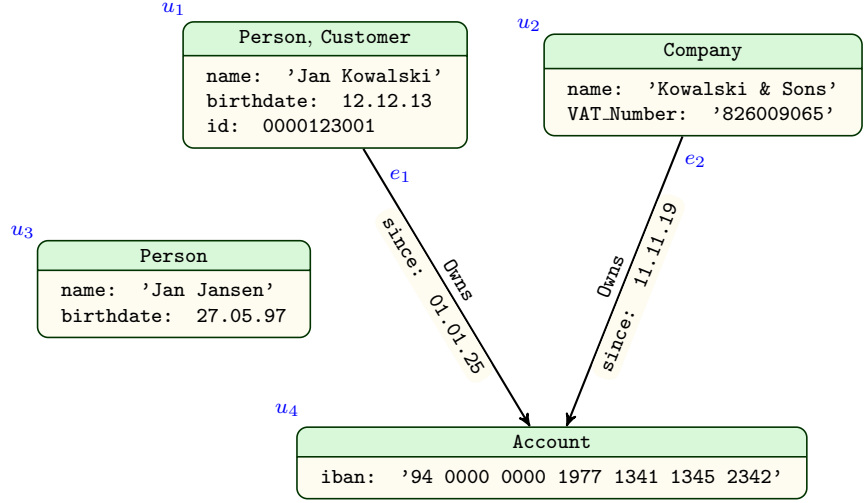

Figure~\ref{fig:graph} shows graph $G=(N,E,\rho, \lambda, \pi)$ where
$N=\{u_1, u_2, u_3, u_4\}$, $E=\{e_1, e_2\}$, $\rho(e_1) = (u_1,u_4)$, $\rho(e_2) = (u_2,u_4)$, $\lambda(u_1) = \{{\small\texttt{Person}}, {\small\texttt{Customer}}\}$, $\lambda(u_2) = \{\small{\texttt{Company}}\}$, $\lambda(u_3) = \{\small{\texttt{Person}}\}$, $\lambda(u_4) = \{\small{\texttt{Account}}\}$, $\lambda(e_1) = \lambda(e_2) =\{\small{\texttt{Owns}}\}$, and $\pi(u_1) = \{{\small \texttt{name}}\mapsto {\small\texttt{'Jan Kowalski'}},\, {\small\texttt{birthdate}} \mapsto {\small\texttt{12.12.13}},\,{\small\texttt{id}} \mapsto {\small\texttt{0000123001}}\}$, etc. 

For a graph $G = (N, E, \rho, \lambda, \pi)$ we write $\|G\|$ for the total number of nodes, edges, labels, and properties in $G$; that is,  \[\|G\| = |N|+|E|+\sum_{x\in N\cup E} \big(|\lambda(x)| + |\pi(x)|\big)\,.\] 
For $v\in N$ we let $\|v\|_G=|\lambda(v)| + |\pi(v)|$,  and for $e\in E$ we let $\|e\|_G=|\lambda(e)| + |\pi(e)| + \|u_1\|_G + \|u_2\|_G$ where $\rho(e)=(u_1, u_2)$.

\section{PG-Schema}
\label{sec:pgschema-def}

\begin{figure}[t]
\centering
    \footnotesize
    \begin{tikzpicture}[
      Node/.style={
        rounded corners,
        rectangle split,
        rectangle split parts=2,
        rectangle split part fill={green!80!black!15!white,solarized-base3!50!white},
        draw=green!20!black,
        semithick
      },
      ISA/.style={
        thick,-{Triangle[open]},
      },
      Edge/.style={
        thick,-stealth',
      },
      EdgeProperties/.style={
        sloped,
        below,
        rounded corners,
        fill=solarized-base3!50!white,
      },
      Junction/.style={
        circle,
        draw=black,
        fill=white,
        thick
      }
      ]


      \node[Node] (person) at (1, 3.3) {
        \texttt{Person}
        \nodepart{two}
        \begin{tabular}{l}
        \texttt{name STRING}\\
        \texttt{birthdate DATE}
        \end{tabular}
      };

      \node[Node] (company) at (5, 3.3) {
        \texttt{Company}
        \nodepart{two}
        \begin{tabular}{l}
        \texttt{name STRING}\\

        \texttt{VAT\_Number STRING}
        \end{tabular}
      };

      \node[Node] (customer) at (1, 1) {
        \texttt{Customer}
        \nodepart{two}
        \begin{tabular}{l}
        \texttt{id INT32 PK}
        \end{tabular}
      };
      
      \node[Node] (account) at (7, 1) {
        \texttt{Account}
        \nodepart{two}
        \begin{tabular}{l}
        \texttt{iban STRING PK}
        \end{tabular}
      };

      \draw[] (customer) edge[ISA] 
      (person);

      \draw[] (1,2) edge (5,2);
      \draw[] (5,2) edge[ISA] 
      (company);

      \draw[] (customer) edge[Edge] 
      node[above] {\texttt{Owns}} 
      node[below,EdgeProperties] {\texttt{since DATE}}
      (account);

      \node[Junction] (xor) at (1,2) {\scriptsize{\texttt{|}}};

      \node at (2.6,1.26){\texttt{1..*}};

      \node at (5.1,1.26){\texttt{0..1}};
     
    \end{tikzpicture}
\caption{A diagram of a customer graph schema.}
\label{fig:schema-diagram}
\end{figure}
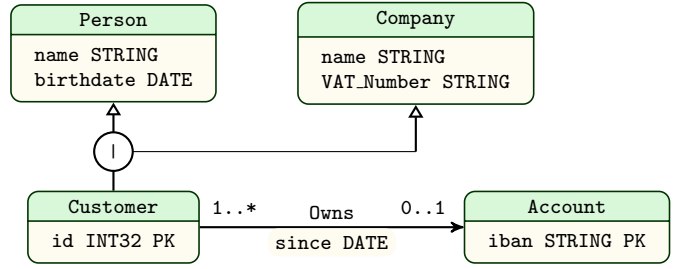

\begin{figure}[t!]
\begin{lstlisting}
CREATE GRAPH TYPE customerGraph STRICT {
  (person: Person {name STRING, birthdate DATE}),
  (company: Company {name STRING, VAT_Number STRING}),
  (customer: Customer & (person|company) {id INT32}),
  (account: Account {iban STRING}),
  (:customer) -[owns: Owns {since DATE}]-> (:account),
  FOR (c:customer) EXCLUSIVE c.id,
  FOR (a:account) EXCLUSIVE a.iban,
  FOR (a:account) MANDATORY ()-[:owns]->(a),
  FOR (c:customer) SINGLETON (c)-[:owns]->()
}

\end{lstlisting}
\caption{\pgschema{} of a customer graph schema.}
\label{fig:pg-schema}
\end{figure}

\pgschema \cite{pgschema} is a formalism for specifying graph schemas, called \emph{graph types} by means of \emph{node} and \emph{edge types}, as well as \emph{integrity constraints} expressed in a sublanguage called \pgkeys~\cite{pgkeys}. We explain the core of \pgschema{} on an example; for full syntax and semantics see \cite{pgkeys,pgschema}.

Consider the graph schema shown in Figure~\ref{fig:schema-diagram} as an ERD-like diagram. The branching generalisation arch decorated with $|$ means that each customer is either a person or a company. This graph schema can be expressed in \pgschema as shown in Figure~\ref{fig:pg-schema}. The schema defines a graph type \lstinline{customerGraph} with node types \mbox{\lstinline{person},} \lstinline{company}, 
\lstinline{customer}, and
\lstinline{account}, edge type \lstinline{owns}, and four integrity constraints. The keyword \lstinline{STRICT} indicates that in a graph of type \lstinline{customerGraph}, each node and edge must conform to one of the defined types. The other alternative is \lstinline{LOOSE}, which indicates that nodes and edges might not conform to any of the defined types, but integrity constraints must still hold.

Each node type $t$ specifies possible combinations of labels and properties. If a node $u$ satisfies this specification, we say $u$ \emph{conforms to type} $t$. A node \emph{of type $t$} is a node that conforms to $t$. For example, a node conforms to type \lstinline{person} if it has label \mbox{\lstinline{Person},} a property \lstinline{name} of type \lstinline{STRING}, and a property \lstinline{birthdate} of type \lstinline{DATE}. By default, types are \emph{closed}; that is, no other labels or properties are allowed. To change this,  we use the keyword \lstinline{OPEN}; for instance, 
\begin{lstlisting}
(person: Person OPEN {name STRING, birthdate DATE, OPEN})  
\end{lstlisting}
allows both additional labels and additional properties. \pgschema allows building more complex types from simpler ones by means of \emph{type union} \lstinline{|} and \emph{type combination} \lstinline{&}.  For example, a node of type \lstinline{customer} must have label \lstinline{Customer} and property \lstinline{id} of type \lstinline{INT32}, as well as either all labels and properties prescribed by type \lstinline{person} or all labels and properties prescribed by type \lstinline{company}. 

Edge types specify edge labels and properties, as well as the types of the endpoints. As for nodes, if an edge (along with its endpoints) satisfies the specification, we say it \emph{conforms} to the type, and an edge \emph{of type} $t$ is an edge that conforms to $t$. For instance, an edge of type \lstinline{owns} must have label \lstinline{Owns} and property \lstinline{since} of type \lstinline{DATE}, and it must lead from a node of type \lstinline{customer} to a node of type \lstinline{account}.

There are three basic kinds of integrity constraints. \lstinline{EXCLUSIVE} constraints ensure that indicated values never repeat for objects within a specified scope. For instance, in Figure~\ref{fig:pg-schema}, nodes of type \lstinline{customer} have different value of the \lstinline{id}
property. \lstinline{MANDATORY} constraints ensure that at least one specified object exists for each object within a specified scope. For instance, in Figure~\ref{fig:pg-schema}, each node of type \lstinline{account} has at least one incoming edge of type \lstinline{owns}.
Dually, \lstinline{SINGLETON} constraints ensure that at most one specified object exists for each object within a specified scope. For instance, in Figure~\ref{fig:pg-schema}, each node of type \lstinline{customer} has at most one outgoing edge of type \lstinline{owns}. \lstinline{MANDATORY} and \lstinline{SINGLETON} constraints are special cases of \lstinline{AT LEAST} $k$ and \lstinline{AT MOST} $k$ constraints, which are also supported in $\pgschema$.

In the remainder, by a (graph) schema, we mean a \pgschema graph type. For a schema $S$, we say that a graph \emph{conforms to  $S$} iff it is of type $S$. 




\section{Validation Problems}
\label{sec:validation-problems}

Validation of property graphs against \pgschema comprises several subtasks, corresponding to different stages of the validation process.

The most fundamental task is \emph{node type conformance}: given a node $u$ of a graph $G$ and a node type $t$ from a schema $S$, we must decide whether $u$ conforms to $t$. For instance, in our running example, for node $u_1$ and type \lstinline{person} we should answer `no', because $u_1$ has label \lstinline{Customer} and property \lstinline{id}, which are not allowed by type \lstinline{person}. For node $u_1$ and type \lstinline{customer} we should answer `yes'. In the context of node type conformance, complexity is measured in terms of $\|u\|_G$ rather than $\|G\|$, because only access to the labels and properties of $u$ is needed. More precisely, data complexity is measured exclusively in terms of $\|u\|_G$  and combined complexity is measured in terms of $\|u\|_G$  and the size $\|S\|$ of schema $S$.

Similarly, in the \emph{edge type conformance} problem,  given an edge $e$ in a graph $G$ and an edge type $t$ from a schema $S$, we must decide whether $e$ conforms to $t$. Note that this involves checking if the endpoints of the edge conform to node types specified in the edge type $t$, which means that node type conformance is a subtask of edge type conformance.  Continuing the running example, for edge $e_1$ and type \lstinline{owns} we should answer `yes', but for edge $e_2$ and type \lstinline{owns} we should answer `no', because the source of $e_2$ is not of type \lstinline{customer} as it is missing label 
\lstinline{Customer} and property 
\lstinline{id}. 
In the context of edge type conformance, we measure data complexity in terms of $\|e\|_G$ and combined complexity in terms of 
$\|e\|_G$  and $\|S\|$.

The third problem is \emph{constraint checking}: given a graph $G$ and an integrity constraint $C$ in schema $S$, we have to decide if  $C$ holds in $G$. In the running example, the answer would be `yes' for all four constraints. Because constraints in $\pgschema$ may refer to node and edge types (e.g., in the scope), node type conformance and edge type conformance are subtasks of constraint checking.  
Unlike node and edge type conformance, constraint checking requires access to the whole graph $G$, which is why we measure data complexity in terms of $\|G\|$ and combined complexity in terms of $\|G\|$ and $\|S\|$. 

Finally, in the problem \emph{graph type conformance}, we are given a graph $G$ and a schema $S$, and need to decide it $G$ conforms to $S$. Recall that for strict graph types this involves checking if each node and edge in $G$ conforms to one of the types listed in $S$ and each integrity constraint in $S$ holds in $G$, and for loose schemas we only need to check if
each integrity constraint in $S$ holds in $G$. In the running example, nodes $u_1$, $u_2$, $u_3$, and $u_4$ conform to types \lstinline{customer}, \lstinline{company}, \lstinline{person}, and \lstinline{account}, respectively,  and edge $e_1$ conforms to type \lstinline{owns}, but edge $e_2$ does not conform to  \lstinline{owns}, which is the only edge type listed in the schema. Hence, the answer should be `no'. However, if we replaced \lstinline{STRICT} with \lstinline{LOOSE} in the defnition of the graph type,  the answer should be `yes', because all four integrity constraints are satisfied. Let us point out  that graph type conformance has all three previously defined problems as subtasks. As for constraint checking, we measure data complexity  in terms of $\|G\|$ and combined complexity in terms of $\|G\|$ and $\|S\|$.

For the purpose of the complexity analysis we restrict our attention to graph types that do not use integrity constraints. The reason for this is that \pgkeys are defined  relative to a host query language and both the scope and the value of a key can be specified using an arbitrarily complex query in the host language \cite{pgkeys}. This means that the complexity of the validation problem is entirely dependent on the choice of the host language.  Graph types without integrity constraints, on the other hand, are a fully specified formalism and we can make absolute statements about their complexity.

\section{Normalized types}

An \emph{atomic node type} is a node type of one of the following forms:
\begin{itemize}
\item \lstinline+{}+ -- no labels and no properties, 
\item \lstinline+Lab {}+ -- only label \lstinline{Lab} and no properties, \item \lstinline+OPEN {}+ -- arbitrary labels and no properties,
\item \lstinline+{prop TP}+ -- only property \lstinline{prop} of type \lstinline{TP} and no labels, 
\item \lstinline+{OPEN}+ -- arbitrary properties and no labels. 
\end{itemize}
We will refer to them as \emph{blank},  \emph{label}, \emph{label wildcard}, \emph{property}, and \emph{property wildcard}.

Every node type can be represented as a normalized expression of linear size build from atomic node types and type references (i.e., names of other types) using operators \lstinline{&} and \lstinline{|}. This is done by replacing more complex types by combinations of suitable labels, properties, and wildcards. For example, the node types \lstinline{person} and \lstinline{customer} in the graph type in Figure~\ref{fig:pg-schema} can be defined with normalized expressions
\begin{lstlisting}
    Person {} & {name STRING} & {birthdate DATE}(*\,,*)
    Customer {} & (person | company) & {id INT32}(*\,.*)
\end{lstlisting}
A normalized node type expression can be seen as a tree in which each internal node is labelled with \lstinline{&} or \lstinline{|} and each leaf is labelled with an atomic node type or a type name. 

Edge types can be normalized in a similar way, based on the following \emph{atomic edge types}:
\begin{itemize}
\item \lstinline+()-[]->()+; 
\item \lstinline+(Lab {})-[]->()+, \lstinline+()-[Lab {}]->()+, \lstinline+()-[]->(Lab {})+;
\item \lstinline+(OPEN {})-[]->()+, \lstinline+()-[OPEN {}]->()+, \lstinline+()-[]->(OPEN {})+;
\item \lstinline+({prop TP})-[]->()+, \lstinline+()-[{prop TP}]->()+, \lstinline+()-[]->({prop TP})+;
\item \lstinline+({OPEN})-[]->()+, \lstinline+()-[{OPEN}]->()+, \lstinline+()-[]->({OPEN})+. 
\end{itemize}
We refer to them as \emph{edge blank}, \emph{source label}, \emph{edge label}, \emph{target label}, \emph{source label wildcard}, \emph{edge label wildcard}, \emph{target label wildcard}, etc. 
Note that this normalisation eliminates node type references from  edge types.

We will think of graph types as collections of pairs $(\nu, t)$ where $\nu$ is a node (resp. edge) type name and $t$ is a normalized node (resp. 
edge) type expression. We say a type name $\nu$ is \emph{defined} in a graph type $T$ if $T$ contains a pair $(\nu, t)$ for some  $t$; we then call $t$ the \emph{definition} of $\nu$. An occurrence of a type name in a type expression is called a \emph{reference}. All type names referenced in a graph type must be defined and type references must be acyclic. We only consider strict graph types, because loose graph types without integrity constraints are trivial. We sometimes blur the distinction between type name $\nu$ and its definition $t$ in $T$, and speak simply of a type $t$ in $T$, especially when discussing conformance. We write $|T|$ for the number of types defined in $T$ and $\|T\|$ for the total size of $T$ defined as $\sum_{(\nu,t)\in T} 1+ \|t\|$, where $\|t\|$ is the size of the type expression $t$.

\section{Expanding types}

Consider a graph type $T$. With every type expression  $t$ we can associate its $T$-expansion $\hat t$, which is a type expression obtained from $t$ by plugging in type definitions from $T$, exhaustively: pick a leaf in $t$ that is labelled with a type name $\nu'$ defined in $T$ and substitute it with the definition $t'$ of $\nu'$ in $T$. Because type references are acyclic, this process terminates. Provided that all type names referenced in $t$ are defined in $T$, the resulting type expression $\hat t$ has no more type references. 
The $T$-expansion $\hat t$ is equivalent to $t$, in the sense that exactly the same nodes or edges conform to $t$ and $\hat t$, but it can be exponentially larger. Because of the potential blow-up, algorithms should avoid materializing expansions.

\section{Types in disjunctive form}
\label{ssec:disjunctive}

We call node types $t$ and $t'$ \emph{equivalent}, written as $t\equiv t'$ if exactly the same nodes conform to them. Similarly for edges. 

Up to equivalence, the operators \lstinline{&} and \lstinline{|} are associative and commutative, and \lstinline{&} distributes over \lstinline{|}:  
\begin{gather*} 
t_1 \text{\lstinline{&}} (t_2 \text{\lstinline{&}} t_3) \equiv (t_1 \text{\lstinline{&}} t_2) \text{\lstinline{&}} t_3 \quad  t_1 \text{\lstinline{|}} (t_2 \text{\lstinline{|}} t_3) \equiv (t_1 \text{\lstinline{|}} t_2) \text{\lstinline{|}} t_3 \\
t_1 \text{\lstinline{&}} t_2 \equiv t_2 \text{\lstinline{&}} t_1 \quad t_1 \text{\lstinline{|}} t_2 \equiv t_2 \text{\lstinline{|}} t_1 \quad t_1 \text{\lstinline{&}} (t_2 \text{\lstinline{|}} t_3) \equiv (t_1 \text{\lstinline{&}} t_2) \text{\lstinline{|}} (t_1 \text{\lstinline{&}} t_3) 
\end{gather*}

We call a node or edge type  \emph{conjunctive} if it is built from atomic types using only \lstinline{&}. From the equivalences above it follows immediately that every conjunctive type $s$ can be equivalently written as a flat combination $a_1 \text{\lstinline{&}} a_2 \text{\lstinline{&}} \dots \text{\lstinline{&}} a_k$ of pairwise different atomic types $a_1, a_2, \dots, a_k$ for some $k \leq \|s\|$. Similarly, every node or edge type $t$ can be equivalently written in the \emph{disjunctive form} as a flat union $s_1 \text{\lstinline{|}} s_2 \text{\lstinline{|}} \dots \text{\lstinline{|}} s_n$  of at most exponentially many conjunctive types $s_1, s_2, \dots, s_n$, each of size at most  $\|t\|$.

\section{Data complexity of \pgtypes}

\begin{proposition}
\label{prop:datacomplexity}
Node, edge, and graph type conformance with respect to a fixed graph type can be decided in polynomial time.  
\end{proposition} 

\begin{proof} We begin from node type conformance. Let $T$ be a fixed graph type. Consider type $t$ from $T$. Let $\hat t$ be the $T$-expansion of $t$. As discussed in Section~\ref{ssec:disjunctive}, we can rewrite $\hat t$ as a union $s_1 \text{\lstinline{|}} s_2 \text{\lstinline{|}} \dots \text{\lstinline{|}} s_n$ of  conjunctive types $s_1, s_2, \dots, s_n$. Because  $T$ is fixed, the size of the resulting type expression is bounded by a constant. In order to check if a given node $u$ in graph $G$ conforms to $t$ it suffices to check if $u$ conforms to any $s_i$, which is easy because $s_i$ is a conjunctive type. Indeed, node $u$ conforms to a conjunctive type $s$ iff 
\begin{itemize}
\item for each atomic type \lstinline+Lab {}+ included in $s$, $u$ has label \lstinline+Lab+;
\item  for each label \lstinline{Lab} in $u$, $s$ includes the atomic type \lstinline+Lab {}+ or the atomic type \lstinline+OPEN {}+; 
 
\item for each atomic type \lstinline+{prop TP}+ included in $s$, $u$ has property \lstinline{prop} with a value of type \lstinline+TP+; and 
\item for each property \lstinline{prop} with value $d$ in $u$, $s$ includes 
 the atomic type  \lstinline+{prop TP}+ such that $d$ is of type \lstinline+TP+ or the atomic type
\lstinline+{OPEN}+.
\end{itemize}
These conditions can be easily checked in polynomial time.

The algorithm for edge type conformance is analogous. For graph type conformance we simply iterate over all nodes in $G$ and node types in $T$, and over all edges in $G$  and edge types in $T$.
\end{proof}

\section{Pruning graph types}
\label{ssec:pruning}

Given a node $u$ in graph $G$, and a  graph type $T$, we can efficiently \emph{prune $T$ wrt.~$u$}  to obtain graph type $T_u$ such that 
\begin{itemize}
\item 
$T_u$ does not use atomic node types that are \emph{incompatible} with $u$: \lstinline+Lab {}+ such that $u$ has no label \lstinline+Lab+ and \lstinline+{prop TP}+ such that $u$ has no property \lstinline+prop+ with value of type \lstinline+TP+;
\item $T_u$ defines a subset of types defined by $T$ and for every type $t$ in $T$, $u$ conforms to $t$ in the context of $T$ iff $t$ is defined in $T_u$ and $u$ conforms to $t$ in the context of $T_u$.
\end{itemize}

This is achieved by eliminating from type expressions in $T$ all leaves labelled with atomic types incompatible with $u$. However, when a child of an \lstinline{&}-node is removed, the \lstinline{&}-node must be removed too; if all children of an \lstinline{|}-node are removed, the \lstinline{|}-node must be  removed too; and when the root of the expression defining type $t$ is removed, all references to $t$ must be removed too.


Similarly, given an edge $e$ in graph $G$, one can prune graph type $T$ wrt. $e$, without affecting the conformance of $e$, to obtain graph type  $T_e$ that only uses atomic edge types compatible with $e$: 
\begin{center}
\lstinline+(Lab {})-[]->()+, \quad \lstinline+()-[Lab {}]->()+, \quad \lstinline+()-[]->(Lab {})+
\end{center}
such that the source of $e$, $e$, and the target of $e$, respectively, has label \lstinline{Lab}, and 
\begin{center}
\lstinline+({prop TP})-[]->()+, \quad \lstinline+()-[{prop TP}]->()+, \quad \lstinline+()-[]->({prop TP})+
\end{center}
such that the source of $e$, $e$, and the target of $e$, respectively, has propety  \lstinline{prop} with value of type \lstinline{TP}. 

Pruning can be done in polynomial time. With a little bit of care it can be  implemented in quadratic time: $O\big(\|T\|\cdot \big(|T| +  \|u\|\big)\big)$ for $T_u$ and $O\big(\|T\| \cdot\big(|T| + \|e\|\big)\big)$ for $T_e$. 

\section{Witnessing branch-sets}

A \emph{branch-set} in a type expression $t$ is a set of branches in the expression tree of $t$ such that no two branches split in a node labelled with \lstinline{|}.  A branch-set $B$ in a node type expression $t$ \emph{covers} a node $u$ in $G$ if 
\begin{itemize}
\item for each label \lstinline{Lab} in $u$, some branch in $B$ ends in a leaf labelled with \lstinline+Lab {}+ or \lstinline+OPEN {}+;
\item for each property \lstinline{prop} with value $d$ in $u$, some branch in $B$ ends in a leaf labelled with \lstinline+{OPEN}+  or with \lstinline+{prop TP}+ such that $d$ is of type \lstinline{TP}. 
\end{itemize}

Similarly, a branch-set $B$ in an edge type expression $t$ covers an edge $e$ in graph $G$ if for each label and property in $e$ or one of its endpoints, some branch in $B$ leads to a leaf labeled with a suitable atomic edge type. 

\begin{lemma}
\label{lem:witness}
Let $G$ be a graph and $T$ a graph type. 
A node $u$ in $G$ conforms to a node type $t$ in  $T_u$ iff $u$ is covered by a branch-set of size at most  $\|u\|_G$ in the $T_u$-expansion of $t$. The same holds for edges. 
\end{lemma}

\section{Combined complexity of \pgtypes}

\begin{theorem}
\label{thm:combinedcomplexity}
Node, edge, and graph type conformance are $\np$-complete problems (in terms of combined complexity).  
\end{theorem}

\begin{proof}
We begin with the upper bound for node type conformance. Let $u$ be a node in $G$ and $t$ a type in graph type $T$. We can prune $T$ in polynomial time to obtain $T_u$. By Lemma~\ref{lem:witness} it is enough to check if there is branch-set of size at most $\|u\|_G$ in the expansion of $t$ that covers $u$. The idea is to guess this branch-set and verify that it covers $u$. Owing to the acyclicity of type references, every branch has length bounded by the size of $T$, so the total size of the branch-set is polynomial. The only obstacle is that we cannot compute the expansion of $t$, as it might have exponential size. Instead, we trace branches from the root of $t$  moving from one type expression to another through type references. Checking that no branches split in an \lstinline{|}-node requires a bit of care, because branches might split in an \lstinline{&}-node in some type expression, then pass through a reference to the same type $t'$, continue together, until they maybe split again, possibly in an \lstinline{|}-node  (only the first split matters). Checking that a guessed branch-set covers $u$ is straightforward.

The upper bound for edge type conformance is entirely analogous. For graph type conformance, we compute the pruned graph type for each node and edge in the input graph, guess a type in the pruned graph type, guess a covering branch-set, and verify it. We accept if each branch-set passes the verification.

Next we prove the lower bound for node type conformance. We reduce from the \textsc{CNF-SAT} problem. Let $\varphi = C_1 \land C_2 \land \dots \land C_m$ be a propositional formula over variables $x_1,x_2, \dots, x_n$, where each clause $C_i$ is a disjunction of literals (variables or their negations); e.g., $(x_1 \lor \lnot x_2 \lor x_3) \land (\lnot x_1 \lor x_2 \lor  x_3)\land(\lnot x_3)$. We construct a node type expression $t_\varphi$ and graph $G_\varphi$  consisting of a single node $u$ such that $u$ conforms to $t_\varphi$ iff $\varphi$ is satisfiable. 
We use $C_1, C_2, \dots, C_m$ as labels. We let $t_\varphi = (t_{x_1} \mathbin{\texttt{|}} t_{\lnot x_1}) \mathbin{\texttt{\&}} (t_{x_2} \mathbin{\texttt{|}} t_{\lnot x_2}) \mathbin{\texttt{\&}} 
\ldots \mathbin{\texttt{\&}} (t_{x_n} \mathbin{\texttt{|}} t_{\lnot x_n})$ where $t_L$ is the combination of types $C_i\,\texttt{\{\}}$ such that clause $C_i$ contains literal $L$; if there are none, we let $t_L=\texttt{\{\}}$. For the formula above, the type expression is 
\[
\big(
C_1\,\texttt{\{\}} \mathbin{\texttt{|}}
C_2\,\texttt{\{\}} 
\big)
\mathbin{\texttt{\&}} 
\big(
C_2\,\texttt{\{\}} 
\mathbin{\texttt{|}} 
C_1\,\texttt{\{\}} 
\big)
\mathbin{\texttt{\&}} 
\big((
C_1\,\texttt{\{\}} 
\mathbin{\texttt{\&}} 
C_2\,\texttt{\{\}} 
)
\mathbin{\texttt{|}}
C_3\,\texttt{\{\}}
\big).
\]
The graph $G_\varphi$ contains a single node $u$ with labels $C_1, C_2, \ldots, C_m$ and no properties. It is routine to verify that the reduction is correct. 

Note also that $\varphi$ is satisfiable iff $G_\varphi$ conforms to the (strict) graph type whose only node type is $t_\varphi$. This shows that graph conformance is $\np$-hard, too. For edge type conformance we replace atomic node types $C_i\,\texttt{\{\}}$ with atomic edge types $\texttt{()-[}C_i\,\texttt{\{\}]->()}$ in the definition of $t_\varphi$ and let $G_\varphi$ be a graph with two nodes without any labels or properties, and a single edge between them with labels $C_1, C_2, \dots, C_m$ and no properties.
\end{proof}

\section{A tractable case}

The type expression constructed in the reduction showing $\np$-hardness of type conformance in Theorem~\ref{thm:combinedcomplexity} is a combination of unions of combinations of atomic types. As it turns out, this nesting pattern of \lstinline{&} and \lstinline{|} is precisely the source of hardness. 

A type $t$ in a graph type $T$ is \emph{\lstinline{&|&}-free} if  its $T$-expansion contains no \lstinline{|}-labelled node that has an  \lstinline{&}-labelled ancestor and an \lstinline{&}-labelled descendant. A graph type $T$ is \emph{\lstinline{&|&}-free} if each type in $T$ is \emph{\lstinline{&|&}-free}. Note that the pruning procedure described in  Section~\ref{ssec:pruning} preserves \lstinline{&|&}-freeness.

\begin{theorem}
Type conformance for \lstinline{&|&}-free types  is in $\ptime$.  
\end{theorem}

\begin{proof} It suffices to prove the claim for node and edge type conformance: the case of graph type conformance  follows immediately from these two.   

We begin with a special case. Consider a type $t = t_1 \verb+&+ t_2 \verb+&+ \dots \verb+&+ t_n$ where each  $t_i$ is a union of atomic node types: labels and properties, but not wildcards. Suppose also that $t$ is pruned with respect to node $u$ in graph $G$. We can decide whether $u$ conforms to $t$ by reduction to the maximum matching problem. Define a bipartite graph whose nodes are $t_1, \dots, t_n$ and all labels and properties of $u$.  We include an edge from $t_i$ to a label \lstinline+Lab+ if   \lstinline+Lab {}+ is one of the disjuncts of $t_i$. Similarly, we include an edge from $t_i$ to a property \lstinline+prop+ if one of the disjuncts of $t_i$ is  \lstinline+{prop TP}+ such that the value of \lstinline+prop+ in $u$ is of type \lstinline+TP+. Clearly, $u$ conforms to $t$ iff the bipartite graph has a matching of size $\|u\|$, which must be the maximum matching, as there are no more nodes in the second part of the graph. Using the Hopcroft--Karp algorithm~\cite{HopcroftK73} we can compute the maximum matching in a bipartite graph with $V$ nodes and $E$ edges in time $O(E\sqrt{V})$. In our case, $V=n+\|u\|$ and $E=|t|$. The argument for edge type conformance is identical, except that in the construction of the bipartite graph we must make sure that labels and properties of the source node, the target node and the edge itself are represented separately.

Suppose now that $t$ does contain wildcards. We eliminate them one by one. Pick a wildcard atomic type $a$ (in node case there are only two: label wildcard and property wildcard; in the edge case there are six). Branch out into the following cases.  
\begin{itemize}
\item  Remove all occurrences of $a$ from $t$. If some $t_j$ becomes empty, remove it from $t$. Solve the remaining problem recursively. 
\item Pick $i$ such that $t_i$ contains $a$ as a disjunct and remove $t_i$ from $t$. Remove $a$ from the remaining $t_j$'s  (if $t_j$ becomes empty, remove it from $t$). Remove all components of the input node or edge that are covered by $a$: if $a$ is \lstinline+OPEN {}+ remove all labels, if $a$ is \lstinline+{OPEN}+ remove all properties, and similarly in the edge case. Solve the remaining problem recursively.  
\end{itemize}
When the last wildcard atomic type is eliminated, use the previously described method for types without wildcards.  

This algorithm has linear branching and the depth of the recursion is bounded by the number of different wildcard atomic types: two if $t$ is a node type and six if $t$ is an edge type. This means that the algorithm is polynomial.

Let us now consider the general case. We focus on node type conformance; the argument for edges is identical. Let $t$ be a \lstinline{&|&}-free node type in graph type $T$ and let $u$ be a node in graph $G$. Without loss of generality we can assume that $T$ is pruned with respect to $u$. We can represent the $T$-expansion $\hat t$ of $t$ as a DAG of size linear in the size of $t$ and $T$, simply by materializing each type reference as an edge to the root of the suitable defining expression in $T$ (references to the same type become edges with the same target). The DAG has a single source node $x$, corresponding to the root of the expression defining type $t$ in $T$ and multiple sink nodes. 

Because $t$ is \lstinline{&|&}-free, we can partition the non-sink nodes of the DAG into three layers: top, middle, and bottom. Nodes from the middle layer are labelled with  \lstinline{&}, and nodes from the top and bottom layers are labelled with \lstinline{|}. All non-sink successors of nodes from the bottom layer belong to the bottom layer as well, and all non-sink successors of nodes from the middle layer belong to the middle or bottom layer. By introducing a linear number of dummy nodes we can ensure that all three layers are nonempty, all successors of nodes from the top layer belong to the top or middle layer, and all successors of nodes from the middle layer belong to the middle or bottom layer. In particular, only nodes from the bottom layer can have sink nodes as successors. Every node $z$ from the bottom layer can be \emph{compiled} into a flat union of linearly many atomic types, simply by collecting all sink nodes that are reachable from $z$. Next, we can compile each node $y$ in the middle layer into a combination of linearly many disjunctions of linearly many atomic types: we do so by taking the combination of types associated with nodes from the bottom layer that are reachable from $y$ by a directed path that does not visit any other nodes from the bottom layer. It then suffices to check whether node $u$ conforms to the type thus associated with some node $y$ from the middle layer, which falls within the special case we have already solved. 
\end{proof}

\printbibliography
\end{document}